\documentclass[envcountsame,a4paper]{llncs}
\usepackage{ifthen}
\newcommand{\techRep}{true} 
\newcommand{\iftechrep}{\ifthenelse{\equal{\techRep}{true}}}

\usepackage{amsmath}
\usepackage{amssymb}
\usepackage{stmaryrd}
\usepackage{graphicx}
\usepackage{algorithm}
\usepackage{algorithmic}
\usepackage[center]{subfigure}
\usepackage{gastex,amsfonts,amssymb,latexsym}
\usepackage{color}
\usepackage[all]{xy}
\usepackage{wrapfig}
\usepackage{btran}
\usepackage{tikz}

\newcommand{\A}{\mathcal{A}}
\newcommand{\N}{\mathbb{N}}
\newcommand{\Q}{\mathbb{Q}}

\newcommand{\ftrees}[1]{\llparenthesis #1 \rrparenthesis}
\newcommand{\trees}[1]{\llbracket #1 \rrbracket}
\renewcommand{\Pr}{\textup{Pr}}

\newcommand{\mb}[1]{\mathbb{#1}}

\newcommand{\ms}[1]{\mathsf{#1}}

\newenvironment{qlemma}[1]{%
{\par\mbox{}\newline\noindent\bf Lemma #1.}%
\begin{itshape}%
}{%
\end{itshape}%
}

\newenvironment{qproposition}[1]{%
{\mbox{}\newline\noindent\bf Proposition #1.}
\begin{itshape}%
}{%
\end{itshape}%
}

\title{Model Checking Stochastic Branching Processes\thanks{This
    work was partially supported by the ERC Advanced Grant VERIWARE and EPSRC grant EP/F001096/1.
    Stefan Kiefer is supported by a DAAD postdoctoral fellowship.}}
\author{Taolue Chen \and Klaus Dr{\"a}ger \and Stefan Kiefer}
\institute{University of Oxford, UK
\email{\{taolue.chen,klaus.draeger,stefan.kiefer\}@cs.ox.ac.uk}
}

\begin{document}

\maketitle

\begin{abstract}
%
%
%
Stochastic branching processes are a classical model
for describing random trees, which have applications in numerous fields including biology, physics, and natural language processing. In particular,
they have recently been proposed to describe parallel programs with stochastic process creation.
In this paper, we consider the problem of model checking stochastic branching process.  Given a
branching process and a deterministic parity tree automaton, we
are interested in computing the probability that the generated random tree
is accepted by the automaton.  We show that this probability can be
compared with any rational number in PSPACE, and with $0$ and~$1$ in
polynomial time.  In a second part, we suggest a tree extension of the
logic PCTL, and develop a PSPACE algorithm for model checking a
branching process against a formula of this logic. We also show that
the qualitative fragment of this logic can be model checked in
polynomial time.
\end{abstract}


\section{Introduction} \label{sec:intro}

Consider an interactive program featuring two types of threads: interruptible threads (type~$I$) and blocking threads (type~$B$)
 which perform a non-interruptible computation or database transaction.
An $I$-thread responds to user commands which occasionally trigger the creation of a $B$-thread.
A $B$-thread may either terminate, or continue, or spawn another $B$-thread in an effort to perform its tasks in parallel.
Under probabilistic assumptions on the thread behaviour, this scenario can be modelled as a \emph{stochastic branching process} as follows:
\begin{align}
 &I \btran{0.9} I     && B \btran{0.2} D && D \btran{1} D \nonumber \\[-2mm]
 &I \btran{0.1} (I,B) && B \btran{0.5} B \label{eq:intro-example} \\[-2mm]
 &                    && B \btran{0.3} (B,B) \nonumber
\end{align}
This means, e.g., that a single step of an $I$-thread spawns a $B$-thread with probability~$0.1$.
We have modelled the termination of a $B$-thread as a transformation into a ``dead'' state~$D$.%
\footnote{We disallow ``terminating'' rules like $B \btran{0.2} \varepsilon$.
This is in contrast to classical branching processes, but technically more convenient for model checking,
where absence of deadlocked states is customarily assumed.}
A ``run'' of this process unravels an infinite tree whose branches record the computation of a thread and its ancestors.
For example, Figure~\ref{fig:intro-example} shows the prefix of a tree that the example process might create.
The probability of creating this tree prefix is the product of the probabilities of the applied rules,
 i.e., $0.1 \cdot 0.9 \cdot 0.1 \cdot 0.3 \cdot 0.5 \cdot 0.2$.

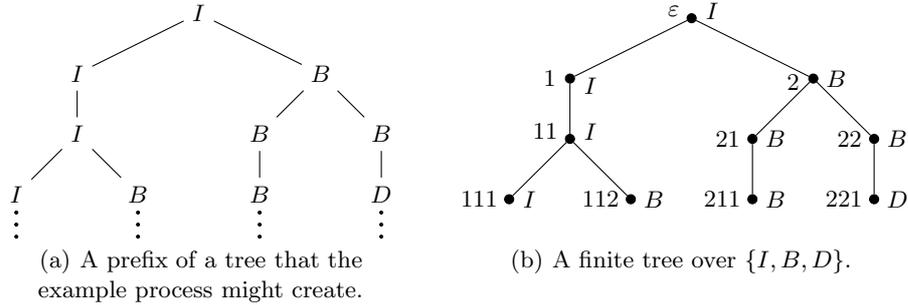
\begin{figure}[t]
\subfigure[A prefix of a tree that the example process might create.\label{fig:intro-example}]{
\begin{tikzpicture}[scale=0.8]
  \node (0) at (0,0) {$I$};
  \node (1) at (-2,-1) {$I$};
  \node (2) at (2,-1) {$B$};
  \node (11) at (-2,-2) {$I$};
  \node (111) at (-3,-3) {$I$};
  \node (112) at (-1,-3) {$B$};
  \node (21) at (1,-2) {$B$};
  \node (22) at (3,-2) {$B$};
  \node (211) at (1,-3) {$B$};
  \node (221) at (3,-3) {$D$};
  \draw (0)--(1);
  \draw (0)--(2);
  \draw (1)--(11);
  \draw (11)--(111);
  \draw (11)--(112);
  \draw (2)--(21);
  \draw (2)--(22);
  \draw (21)--(211);
  \draw (22)--(221);
  \draw[fill] (-3,-3.3) circle (0.7pt);
  \draw[fill] (-3,-3.5) circle (0.7pt);
  \draw[fill] (-3,-3.7) circle (0.7pt);
  \draw[fill] (-1,-3.3) circle (0.7pt);
  \draw[fill] (-1,-3.5) circle (0.7pt);
  \draw[fill] (-1,-3.7) circle (0.7pt);
  \draw[fill] (1,-3.3) circle (0.7pt);
  \draw[fill] (1,-3.5) circle (0.7pt);
  \draw[fill] (1,-3.7) circle (0.7pt);
  \draw[fill] (3,-3.3) circle (0.7pt);
  \draw[fill] (3,-3.5) circle (0.7pt);
  \draw[fill] (3,-3.7) circle (0.7pt);
\end{tikzpicture}
}
\quad
\subfigure[A finite tree over $\{I,B,D\}$.\label{fig:preli-1}]{
\raisebox{2mm}{
\begin{tikzpicture}[scale=0.8,every node/.style={circle,draw,inner sep=1.2pt,fill}]
  \node[label={[yshift=1mm]left:$\varepsilon$},label={[yshift=1mm]right:$I$}] (0) at (0,0){};
  \node[label={[yshift=0mm]left:$1$},label={[yshift=-1mm]right:$I$}] (1) at (-2,-1) {};
  \node[label={[yshift=-0.5mm]left:$2$},label={[yshift=0mm]right:$B$}] (2) at (2,-1) {};
  \node[label={[yshift=1mm]left:$11$},label={[yshift=1mm]right:$I$}] (11) at (-2,-2) {};
  \node[label={[yshift=0mm]left:$111$},label={[yshift=0mm]right:$I$}] (111) at (-3,-3) {};
  \node[label={[yshift=0mm]left:$112$},label={[yshift=0mm]right:$B$}] (112) at (-1,-3) {};
  \node[label={[yshift=0mm]left:$21$},label={[yshift=0mm]right:$B$}] (21) at (1,-2) {};
  \node[label={[yshift=0mm]left:$22$},label={[yshift=0mm]right:$B$}] (22) at (3,-2) {};
  \node[label={[yshift=0mm]left:$211$},label={[yshift=0mm]right:$B$}] (211) at (1,-3) {};
  \node[label={[yshift=0mm]left:$221$},label={[yshift=0mm]right:$D$}] (221) at (3,-3) {};
  \draw (0)--(1);
  \draw (0)--(2);
  \draw (1)--(11);
  \draw (11)--(111);
  \draw (11)--(112);
  \draw (2)--(21);
  \draw (2)--(22);
  \draw (21)--(211);
  \draw (22)--(221);
\end{tikzpicture}
}}\vspace{-2mm}
\caption{Figures for Section \ref{sec:intro} (left) and \ref{sec:preli} (right)}\vspace{-2mm}
\end{figure}

This example is an instance of a (stochastic multitype) branching process,
 which is a classical mathematical model with applications
 in numerous fields including biology, physics and natural language processing, see e.g.~\cite{Harris63,AthreyaNey72}.
In~\cite{KW11:TACAS} an extension of branching processes was introduced to model parallel programs with stochastic process creation.
The broad applicability of branching processes arises from their simplicity:
 each \emph{type} models a class of threads (or tasks, animals, infections, molecules, grammatical structures) with the same probabilistic behaviour.

This paper is about model checking the random trees created by branching processes.
Consider a specification that requires a linear-time property to hold along all tree branches.
In the example above, e.g., we may specify that ``no process should become forever blocking'',
 more formally, ``on all branches of the tree we see infinitely many $I$ or $D$''.
We would like to compute the probability that all branches satisfy such a given $\omega$-regular word property.
Curiously, this problem generalises two seemingly very different classical problems:
\begin{itemize}
\item[(i)]
If all rules in the branching process are of the form $X \btran{p} Y$, i.e., each node has exactly one child,
 the branching process describes a finite-state Markov chain.
Computing the probability that a run of such a Markov chain satisfies an $\omega$-regular property
 is a standard problem in probabilistic verification, see e.g.\ \cite{CourcYann95,Vardi99}.
\item[(ii)]
If for each type~$X$ in the branching process there is only one rule $X \btran{1} \alpha$
 (where $\alpha$ is a nonempty sequence of types),
 then the branching process describes a unique infinite tree.
Viewing the types in~$\alpha$ as possible successor states of~$X$ in a finite nondeterministic transition system,
 the branches in the created tree are exactly the possible runs in the finite transition systems.
Of course, checking if all runs in such a transition system satisfy an $\omega$-regular specification is also a well-understood problem.
\end{itemize}
One could expect that well-known Markov-chain based techniques for dealing with problem~(i) can be generalised to branching processes.
This is \emph{not} the case:
it follows from our results that in the example above, the probability that all branches satisfy the mentioned property is~$0$;%
\footnote{Intuitively, this is because a $B$-thread more often clones itself than dies.}
  however, if the numbers $0.2$ and $0.3$ in~\eqref{eq:intro-example} are swapped, the probability changes from~$0$ to~$1$.
This is in sharp contrast to finite-state Markov chains, where qualitative properties (satisfaction with probability $0$ resp.~$1$)
 do not depend on the exact probability of individual transitions.

The rules of a branching process are reminiscent of the rules of probabilistic pushdown automata (pPDA)
 or the equivalent model of recursive Markov chains (RMCs).
However, the model-checking algorithms for both linear-time and branching-time logics proposed for
 RMCs and pPDAs \cite{EKM04,EYstacs05Extended,EY12journal} do \emph{not} work for branching processes,
  essentially because pPDA and RMCs specify Markov chains, whereas branching processes specify random trees.
Branching processes cannot be transformed to pPDAs, at least not in a straightforward way.
Note that if the rules in the example above are understood as pPDA rules with $I$ as starting symbol,
 then $B$ will never even occur as the topmost symbol.

To model check branching processes, we must leave the realm of Markov chains and
 consider the probability space in terms of tree prefixes~\cite{Harris63,AthreyaNey72}.
Consequently, we develop algorithms that are very different from the ones dealing with the special cases (i) and~(ii) above.
Nevertheless, for qualitative problems (satisfaction with probability $0$ resp.~$1$)
 our algorithms also run in polynomial time with respect to the input models,
  even for branching processes that do not conform to the special cases (i) and~(ii) above.

Instead of requiring a linear-time property to hold on all branches,
 we consider more general specifications in terms of \emph{deterministic parity tree automata}.
In a nutshell, our model-checking algorithms work as follows:
(1) compute the ``product'' of the input branching process and the tree automaton;
(2) reduce the analysis of the resulting product process to the problem of computing the probability that all branches reach a ``good'' symbol;
(3) compute the latter probability by setting up and solving a nonlinear equation system.
Step~(1) can be seen as an instance of the automata-theoretic model-checking approach.
The equation systems of step~(3) are of the form $\vec{x} = \vec{f}(\vec{x})$, where $\vec{x}$ is a vector of variables,
  and $\vec{f}(\vec{x})$ is a vector of polynomials with nonnegative coefficients.
Solutions to such equation systems can be computed or approximated efficiently~\cite{EYstacs05Extended,EKL10:SICOMP,EtessamiSY12:poly-time}.
Step~(2) is, from a technical point of view, the main contribution of the paper;
 it requires a delicate and nontrivial analysis of the behaviour of branching processes.

In Section~\ref{sec:logic} we also consider logic specifications.
We propose a new logic, PTTL, which relates to branching processes in the same manner as the logic PCTL relates to Markov chains.
Recall that PCTL contains formulae such as $[\phi \mathsf{U} \psi]_{\ge p}$ which specifies
 that the probability of runs satisfying $\phi \mathsf{U} \psi$ is at least~$p$.
For PTTL we replace the linear-time subformulae such as $\phi \mathsf{U} \psi$
 with tree subformulae such as $\phi \mathsf{EU} \psi$ or $\phi \mathsf{AU} \psi$,
 so that, e.g., $[\phi \mathsf{EU} \psi]_{\ge p}$  specifies that the probability of trees that have a branch satisfying
  $\phi \mathsf{U} \psi$ is at least~$p$,
  and $[\phi \mathsf{AU} \psi]_{\ge p}$ specifies that the probability of trees all whose branches satisfy $\phi \mathsf{U} \psi$ is at least~$p$.
We show that branching processes can be model checked against this logic in PSPACE, and against its qualitative fragment in polynomial time.

\vspace{2mm}
\noindent\emph{Related work.}
The rich literature on branching processes (see e.g.~\cite{Harris63,AthreyaNey72} and the references therein)
 does not consider model-checking problems.
Probabilistic split-join systems \cite{KW11:TACAS} are branching processes with additional features for process synchronisation and communication.
The paper~\cite{KW11:TACAS} focuses on performance measures (such as runtime, space and work), and does not provide a functional analysis.
The models of pPDAs and RMCs also feature dynamic task creation by means of procedure calls,
 however, as discussed above, the existing model-checking algorithms~\cite{EKM04,EYstacs05Extended,EY12journal} do not work for branching processes.
Several recent works~\cite{EYstacs05Extended,EKL10:SICOMP,EtessamiSY12:poly-time} have studied
 the exact and approximative solution of fixed-point equations of the above mentioned form.
Our work connects these algorithms with the model-checking problem for branching processes.

\vspace{2mm}
\noindent\emph{Organisation of the paper.}
After some preliminaries (Section~\ref{sec:preli}), we present our results on parity specifications in Section~\ref{sec:parity}.
In Section~\ref{sec:logic} we propose the new logic PTTL and develop model-checking algorithms for it.
We conclude in Section~\ref{sec:conclusions}.
\iftechrep{
Some proofs have been moved to an appendix.
}{
Some proofs have been moved to a technical report~\cite{12CDK:mfcs-report}.
}
\section{Preliminaries}  \label{sec:preli}

We let $\mathbb{N}$ and $\mathbb{N}_0$ denote the set of positive and
nonnegative integers, respectively. Given a finite set~$\Gamma$, we
write $ \Gamma^*:= \bigcup_{k \in \N_0}\Gamma^k$ for the set of tuples and
$\Gamma^+ := \bigcup_{k \in \N}\Gamma^k$ for the set of nonempty tuples
over $\Gamma$.

\begin{definition}[Branching process]
A \emph{branching process} is a tuple $\Delta=(\Gamma, \mathord{\btran{}}, \mathit{Prob})$ where $\Gamma$ is a finite set of types,
 $\mathord{\btran{}} \subseteq \Gamma\times \Gamma^+$ is a finite set of transition rules,
 $\mathit{Prob}$ is a function assigning positive probabilities to transition rules
  so that for every $X\in \Gamma$ we have that $\sum_{X \btran{} \alpha}\mathit{Prob}(X \btran{} \alpha)=1$.
\end{definition}
We write $X \btran{p} \alpha$ if $\mathit{Prob}(X \btran{} \alpha) =
p$.  Observe that since the set of transition rules is finite, there
is a global upper bound $K_{\Delta}$ such that $|\alpha| \le
K_{\Delta}$ for all $X \btran{} \alpha$.

A \emph{tree} is a nonempty prefix-closed language $V\subseteq
\mathbb{N}^*$ for which there exists a function $\beta_V:
V\to\N_0$ such that for all $w\in V$ and $k\in\mathbb{N}$,
$wk\in V$ if and only if $k\le\beta_V(w)$. $\beta_V(w)$ is called the
\emph{branching degree} of $w$ in $V$.
We denote by $B_f$ the set of finite
trees, and by $B_i$ the set of infinite trees without leaves
(i.e. trees such that $\beta_V(w) > 0$ for all $w\in V$).
A \emph{prefix} of $V$ is a tree $V'\subseteq V$ such that for
all $w\in V'$,  $\beta_{V'}(w)\in\{0,\beta_V(w)\}$.

A \emph{tree over~$\Gamma$} is a pair $(V,\ell)$ where $V$ is a tree,
 and $\ell: V \to \Gamma$ is a labelling function on the nodes.
Given a tree $t=(V,\ell)$ with a node $u\in V$, we write
$t_u=(V_u,\ell_u)$ for the subtree of~$t$ rooted at~$u$; here
$V_u = \{w\in\mathbb{N}^* \mid uw\in V\}$ and $\ell_u(w) = \ell(uw)$ for
$w\in V_u$.
A tree $(V',\ell')$ is a \emph{prefix} of $(V,\ell)$ if $V'$ is a prefix of $V$ and $\ell'(w) = \ell(w)$ for all $w\in V'$.

A \emph{path} (resp.~\emph{branch}) in a tree~$t=(V,\ell)$ is a finite (resp.~infinite) sequence $u_0, u_1, \ldots$ with $u_i \in V$
 such that $u_0=\epsilon$ is the root of~$t$, and $u_{i+1}=u_ik_i$ for $k_i \in \N$ is a child of~$u_i$.
A \emph{branch label} of~$t$ is a sequence $\ell(u_0), \ell(u_1),\ldots$, where $u_0, u_1, \ldots$ is a branch.
The \emph{successor word} of a node $w\in V$ is $\sigma_t(w) = \ell(w1)\dots\ell(w\beta_V(w))$.

Given a tree~$t = (V,\ell)$ over~$\Gamma$ and a subset $W \subseteq V$, 
 we write $t \models \mathsf{AF} W$ if all its branches
go through~$W$,
i.e., for all $v\in V$ there is a $w\in W$ such that $v$ is a predecessor of~$w$ or vice versa. 
If $\Lambda\subseteq\Gamma$, we write $t\models\mathsf{AF}\Lambda$ for $t\models\mathsf{AF}\{w\in V \mid \ell(w)\in\Lambda\}$.  
Similarly, we write $t \models \ms{AG} \Lambda$ if $\ell(w) \in \Lambda$ for all $w \in V$.

\begin{example} \label{ex:preli-1}
We illustrate these notions.
Figure~\ref{fig:preli-1} shows a finite tree $t = (V,\ell) \in B_f$ over~$\Gamma$ with $\Gamma = \{I,B,D\}$ and $V = \{\varepsilon,1,11,111,112,2,21,211,22,221\}$
 and, e.g.,  $\ell(\varepsilon) = I$ and $\ell(112) = B$.
We have $\beta_V(\varepsilon) = 2$ and $\beta_V(21) = 1$ and $\beta_V(211) = 0$.
The node~$2$ is a predecessor of~$211$.
The tree $t'=(V',\ell')$ with $V' = \{\varepsilon,1,2,21,22\}$ and $\ell'$ being the restriction of~$\ell$ on~$V'$ is a prefix of~$t$.
The sequence $\varepsilon,2,21$ is a path in~$t$.
We have $\sigma_t(11) = IB$.
The tree satisfies $t \models \ms{AF} \{1,21,221\}$ and $t \models \ms{AF} \{I\}$.
\end{example}

A tree $t=(V,\ell)$ over $\Gamma$ is \emph{generated} by a branching
process~$\Delta = (\Gamma, \mathord{\btran{}}, \mathit{Prob})$ if for
every $w\in V$ with $\beta_V(w) > 0$ we have $\ell(w) \btran{} \sigma_t(w)$.  We write
$\ftrees{\Delta}$ and $\trees{\Delta}$ for the sets of trees
$(V,\ell)$ generated by~$\Delta$ with $V\in B_f$ and $V\in B_i$,
respectively. For any $X\in\Gamma$,
$\ftrees{\Delta}_X\subseteq\ftrees{\Delta}$ and
$\trees{\Delta}_X\subseteq\trees{\Delta}$ contain those trees
$(V,\ell)$ for which $\ell(\epsilon) = X$.

\newcommand{\Cyl}{\mathit{Cyl}}
\begin{definition}[Probability space of trees, cf.~\protect{\cite[Chap.~VI]{Harris63}}]
Let $\Delta = (\Gamma, \mathord{\btran{}}, \mathit{Prob})$ be a
branching process. For any finite tree $t=(V,\ell)\in\ftrees{\Delta}$,
let the \emph{cylinder} over $t$ be
$\Cyl_{\Delta}(t):=\{t'\in\trees{\Delta} \mid t\mbox{ is a prefix of
}t'\}$, and define
 $p_{\Delta}(t) := \prod_{w \in V : \beta_V(w)>0}\mathit{Prob}(\ell(w),\sigma_t(w)).$
%
%
For each $X\in\Gamma$ we define a probability space $(\trees{\Delta}_X,\Sigma_X,\Pr_X)$,
 where $\Sigma_X$ is the $\sigma$-algebra generated by $\{\Cyl_{\Delta}(t) \mid t\in\ftrees{\Delta}_X\}$,
  and $\Pr_X$ is the probability measure generated by $\Pr_X(\Cyl_{\Delta}(t)) = p_{\Delta}(t)$.
Sometimes we write $\Pr_X^\Delta$ to indicate~$\Delta$.
We may drop the subscript of $\Pr_X$ if $X$ is understood.
We often write $t_X$ to mean a tree $t \in \trees{\Delta}_X$ randomly sampled according to the probability space above.
\end{definition}

\begin{example}
Let $\Delta = (\Gamma, \mathord{\btran{}}, \mathit{Prob})$ be the branching process with $\Gamma = \{I,B,D\}$ and the rules
 as given in~\eqref{eq:intro-example} on page~\pageref{eq:intro-example}.
The tree~$t$ from Figure~\ref{fig:preli-1} is generated by~$\Delta$: we have $t \in \ftrees{\Delta}_I$.
We have $\Pr_I(\Cyl_{\Delta}(t)) = p_{\Delta}(t) = 0.1 \cdot 0.9 \cdot 0.1 \cdot 0.3 \cdot 0.5 \cdot 0.2$;
this is probability of those trees $t' \in \trees{\Delta}_I$ that have prefix~$t$.
\end{example}
%
%
We say that a quantity $q \in [0,1]$ is \emph{PPS-expressible} if one can compute, in polynomial time,
 an integer $m \in \N$ and a fixed-point equation system $\vec{x} = \vec{f}(\vec{x})$, where $\vec{x}$ is a vector of $m$ variables,
 $\vec{f}$ is a vector of $m$ multivariate polynomials over~$\vec{x}$ with nonnegative rational coefficients,
 $\vec{f}(\vec{1}) \le \vec{1}$ where $\vec{1}$ denotes the vector $(1, \ldots, 1)$,
 and $q$ is the first component of the least nonnegative solution $\vec{y} \in [0,\infty)^m$ of $\vec{x} = \vec{f}(\vec{x})$.
\begin{proposition} \label{prop:PPS-expressible}
 Let $q$ be PPS-expressible. We have:
 \begin{enumerate}
 \item[(a)]
  For $\tau \in \{0,1\}$ one can decide in (strongly) polynomial time whether $q = \tau$.
 \item[(b)]
  For $\tau \in \Q$ one can decide in polynomial space whether
   $q \bowtie \tau$, where
    $\mathord{\bowtie} \in \{ \mathord{<}, \mathord{>}, \mathord{\le}, \mathord{\ge}, \mathord{=}, \mathord{\ne} \}$.
 \item[(c)]
  One can approximate~$q$ within additive error~$2^{-j}$ in time polynomial in both~$j$ and the (binary) representation size of~$\vec{f}$.
 \end{enumerate}
\end{proposition}
Part~(a) follows from~\cite{EYstacs05Extended,EGK10:stacs}.
Part~(b) is shown in~\cite[section~4]{EYstacs05Extended} by appealing to
 the existential fragment of the first-order theory of the reals, which is decidable in~PSPACE, see~\cite{Can88,Renegar92}.
Part~(c) follows from a recent result~\cite[Corollary~4.5]{EtessamiSY12:poly-time}.
The following proposition follows from a classical result on branching processes~\cite{Harris63}.
\begin{proposition} \label{prop:bp-termination}
 Let $\Delta=(\Gamma, \mathord{\btran{}}, \mathit{Prob})$ be a branching process.
 Let $X \in \Gamma$ and $\Lambda \subseteq \Gamma$.
 Then $\Pr[t_X \models \mathsf{AF} \Lambda]$ is PPS-expressible.
\end{proposition}


\section{Parity Specifications} \label{sec:parity}

In this section we show how to compute the probability of those trees 
 that satisfy a given parity specification.

A \emph{(top-down) deterministic (amorphous) parity tree automaton} (DPTA) is a tuple
$\A=(Q,\Gamma,q_0,\delta,c)$, where $Q$ is the finite set of states, $q_0\in Q$ is the initial state,
 $\delta:Q \times \Gamma \times \mathbb{N} \to Q^*$ is the transition function satisfying $|\delta(q,X,n)| = n$ for all $q,X,n$,
 and $c:Q\to\mathbb{N}$ is a colouring function.
Such an automaton~$\A$ maps a tree $t=(V,\ell)$ over~$\Gamma$ to the (unique) tree $\A(t) = (V,\ell')$ over~$Q$ such that
 $\ell(\varepsilon) = q_0$ and for all $w \in V$, $\sigma_{A(t)}(w) = \delta(\ell'(w),\ell(w),\beta_V(w))$.

Automaton~$\A = (Q,\Gamma,q_0,\delta,c)$ \emph{accepts} a tree~$t$ over~$\Gamma$ if
 for all branch labels $q_0 q_1 \cdots \in Q^\omega$ of~$\A(t)$ the highest colour that occurs infinitely often in $c(q_0), c(q_1), \ldots$ is even.
\begin{example}
Recall (e.g., from~\cite{Thomas97handbook}) that any $\omega$-regular \emph{word} property (e.g., any LTL specification)
 can be translated into a deterministic parity \emph{word} automaton.
Such an automaton, in turn, can be easily translated into a DPTA
 which specifies that the labels of \emph{all} branches satisfy the $\omega$-regular word property.
We do not spell out the translation, but let us note that in the resulting tree automaton,
 for all $(q,X) \in Q \times \Gamma$ there is $q' \in Q$ such that $\delta(q,X,k) = (q',\dots,q')$ for all~$k$.
\end{example}

Given a colouring function $c: \Gamma \to \N$, a tree~$(V,\ell)$ over~$\Gamma$ is called \emph{good} for~$c$
 if for each branch $u_0, u_1, \cdots$ the largest number that occurs infinitely often in the sequence
 $c(\ell(u_0)), c(\ell(u_1)), \ldots$ is even.
The following proposition is immediate.
\begin{proposition} \label{prop:product}
 Let $\Delta=(\Gamma, \mathord{\btran{}}, \mathit{Prob})$ be a branching process,
  and let $\A=(Q,\Gamma,q_0,\delta,c)$ be a DPTA.
 Define the \emph{product} of $\Delta$ and~$\A$ as the branching process
  $\Delta_\bullet = (\Gamma \times Q, \mathord{\btran{}_\bullet}, \mathit{Prob}_\bullet)$
   with $(X,q) \btran{p}_\bullet (Y_1,q_1)\dots(Y_k,q_k))$ for $X
   \btran{p} Y_1\dots Y_k$, where $(q_1,\dots,q_k) = \delta(q,X,k)$.
 Define $c_\bullet : \Gamma \times Q \to \N$ by $c_\bullet(X,q) := c(q)$.
 Then for all $X \in \Gamma$ we have
 \[
  \Pr_X^\Delta[t \text{ is accepted by~$\A$}] = \Pr_{(X,q_0)}^{\Delta_\bullet}[t \text{ is good for~$c_\bullet$}]\,.
 \]
\end{proposition}
In view of Proposition~\ref{prop:product}, it suffices to compute the probability $\Pr[t_X \text{ is good for~$c$}]$,
 where a branching process $\Delta=(\Gamma, \mathord{\btran{}}, \mathit{Prob})$ with $X \in \Gamma$ and a colouring function $c: \Gamma \to \N$
  are fixed for the rest of the section. 
We write $\Pr[t_X \text{ is good}]$ if $c$ is understood.
We distinguish between the \emph{qualitative problem},
 i.e., computing whether $\Pr[t_X \text{ is good}] = 1$ holds for a given $X \in \Gamma$,
 and the \emph{quantitative problem}, i.e., the computation of $\Pr[t_X \text{ is good}]$.

\subsection{The Qualitative Problem} \label{sub:parity-qual}

The outline of this subsection is the following:
We will show that the qualitative problem can be solved in polynomial time (Theorem~\ref{thm:parity-qual}).
First we show (Proposition~\ref{prop:qualitative-ltl}) that it suffices to compute all \emph{clean} types, where ``clean'' is defined below.
We will show (Lemma~\ref{lem:cut-or-dirty}) that a type~$X$ is clean if and only if $\Pr[t_X \models \mathsf{AF} \Lambda] = 1$
 holds for suitable set $\Lambda \subseteq \Gamma$.
By Proposition~\ref{prop:bp-termination} the latter condition can be checked in polynomial time,
 completing the qualitative problem.

If there exists a tree $(V,\ell) \in \trees{\Delta}_X$ and a node~$u\in V$ with $\ell(u) = Y$,
 then we say that $Y$ is \emph{reachable} from~$X$.
Given $X \in \Gamma$ and a finite word $w=X_0 \cdots X_m \in \Gamma^+$,
 we say that $w$ is \emph{$X$-closing} if $m \ge 1$ and $X_m = X$ and $c(X_i) \le c(X)$ for $0 \le i \le m$.
A branch with label $X_0 X_1 \cdots \in \Gamma^\omega$ is called \emph{$X$-branch}
 if $X_0 = X$ and there is a sequence $0 = m_0 < m_1 < m_2 < \cdots$ such that
 $X_{m_i} \cdots X_{m_{i+1}}$ is $X$-closing for all $i \in \N$.
We say that a type $Y \in \Gamma$ is \emph{odd} (resp.~\emph{even}), if $c(Y)$ is odd (resp.~even).
Observe that a tree~$t$ is good if and only if for all its vertices~$u$ and all \emph{odd} types~$Y$ the subtree $t_u$ does not have a $Y$-branch.
A type $Y\in \Gamma$ is \emph{clean} if $Y$ is even or $\Pr[t_Y\mbox{ has a $Y$-branch}]=0$.
The following proposition reduces the qualitative problem
 to the computation of all clean types.
\begin{proposition} \label{prop:qualitative-ltl}
We have that $\Pr[t_X \mbox{ is good}]=1$ if and only if all $Y$ reachable from $X$ are clean.
\end{proposition}
\begin{proof}
 If there is an unclean reachable~$Y$, then $\Pr[t_Y \mbox{ is good}] < 1$ and so $\Pr[t_X \mbox{ is good}] < 1$.
 Otherwise, for each node~$v$ in~$t_X$ and for each odd $Y$ we have that $\Pr[(t_X)_v \text{ has a $Y$-branch}] = 0$.
 Since the set of nodes in a tree is countable, it follows that almost surely no subtree of~$t_X$ has a $Y$-branch for odd~$Y$;
  i.e., $t_X$ is almost surely good.
\qed
\end{proof}
Call a path in a tree \emph{$X$-closing} if the corresponding label sequence is $X$-closing.
Given $X\in \Gamma$, we define
 \[N_X := \{Y\in \Gamma \mid \mbox{ no tree in~$\trees{\Delta}_Y$ has an $X$-closing path}\}\,.\]
Note that $c(Y) > c(X)$ implies $Y \in N_X$ and that $N_X$ is computable in polynomial time.
A word $X_0 X_1 \cdots \in (\Gamma^* \cup \Gamma^\omega)$ is called \emph{$X$-failing} if no prefix is $X$-closing and there is $i \ge 0$ with $X_i \in N_X$.
A branch in a tree is called $X$-failing if the corresponding branch label is $X$-failing.
\newcommand{\Clos}{\mathsf{Clos}}%
\newcommand{\Fail}{\mathsf{Fail}}%
Given $X \in \Gamma$ and a tree~$t$, let $\Clos_X(t)$ (resp.~$\Fail_X(t)$) denote the set of those nodes~$w$ in~$t$
 such that the path to~$w$ is $X$-closing (resp.~$X$-failing) and no proper prefix of this path is $X$-closing (resp.~$X$-failing).
We will need the following lemma.
\newcommand{\stmtlemcutorinf}{
 Define the events
  $C := \{t_X \mid t_X \models \mathsf{AF} \left( \Clos_X(t_X) \cup \Fail_X(t_X) \right) \}$
  and $I := \{t_X \mid \Clos_X(t_X) \text{ is infinite}\}$.
 Then $C \cap I = \emptyset$ and $\Pr[C \cup I] = 1$.
}
\begin{lemma} \label{lem:cut-or-inf}
\stmtlemcutorinf
\end{lemma}
The following lemma states in particular that an odd type~$X$ is clean if and only if $\Pr[t_X \models \mathsf{AF} N_X] = 1$.
We prove something slightly stronger:
\newcommand{\lemcutordirty}{
Define the events
 $F := \{t_X \mid t_X \models \mathsf{AF}N_X\}$ and
 $H := \{t_X \mid t_X \text{ has an $X$-branch}\}$.
Then $F \cap H = \emptyset$ and $\Pr[F \cup H] = 1$.
}
\begin{lemma} \label{lem:cut-or-dirty}
\lemcutordirty
\end{lemma}
Now we have:
\begin{theorem} \label{thm:parity-qual}
  One can decide in polynomial time whether $\Pr[t_X \mbox{ is good}]=1$.
\end{theorem}
\begin{proof}
 By Proposition~\ref{prop:qualitative-ltl} it suffices to show that cleanness can be determined in polynomial time.
 By Lemma~\ref{lem:cut-or-dirty} an odd type~$X$ is clean if and only if $\Pr[t_X \models \mathsf{AF} N_X] = 1$.
 The latter condition is decidable in polynomial time by Proposition~\ref{prop:bp-termination}.
\qed
\end{proof}

\begin{example} \label{ex:running}
Consider the branching process with $\Gamma = \{1,2,3,4\}$ and the rules
$1 \btran{1/3} 1 1$, $1 \btran{2/3} 4$, $2 \btran{1/2} 1 3$, $2 \btran{1/2} 2 3$, $3 \btran{2/3} 3 3$, $3 \btran{1/3} 1$, $4 \btran{1} 4$,
 and the colouring function~$c$ with $c(i) = i$ for $i \in \{1,2,3,4\}$.
Using a simple reachability analysis one can compute the sets
$N_1 = \{2,3,4\}$, $N_2 = \{1,3,4\}$, $N_3 = \{1,4\}$, $N_4 = \emptyset$.
Applying Proposition~\ref{prop:bp-termination} we find $\Pr[t_3 \models \mathsf{AF} N_3] < 1 = \Pr[t_1 \models \mathsf{AF} N_1]$.
It follows by Lemma~\ref{lem:cut-or-dirty} that the only \emph{un}clean type is~$3$.
Since type~$3$ is only reachable from~$2$ and from~$3$,
 Proposition~\ref{prop:qualitative-ltl} implies that $\Pr[t_X \text{ is good}] = 1$ holds if and only if $X \in \{1,4\}$.
\end{example}

\subsection{The Quantitative Problem} \label{sub:parity-quant}

Define
 $G := \{X \in \Gamma \mid \text{ all $Y$ reachable from~$X$ are clean}\}
 $.
The following Proposition~\ref{prop:AFG} states that $\Pr[t_X \mbox{ is good}] = \Pr[t_X \models \mathsf{AF} G]$.
This implies, by Proposition~\ref{prop:bp-termination}, that the probability is PPS-expressible
 (see Theorem~\ref{thm:parity-quant}).


\newcommand{\stmtpropAFG}{
We have $\Pr[t_X \mbox{ is good}] = \Pr[t_X \models \mathsf{AF} G]$.
}
\begin{proposition} \label{prop:AFG}
\stmtpropAFG
\end{proposition}
This implies the following theorem.
\begin{theorem} \label{thm:parity-quant}
 For any $X \in \Gamma$ we have that $\Pr[t_X \text{ is good}]$ is PPS-expressible.
\end{theorem}
\begin{proof}
 By Proposition~\ref{prop:AFG} we have $\Pr[t_X \mbox{ is good}] = \Pr[t_X \models \mathsf{AF} G]$.
 So we can apply Proposition~\ref{prop:bp-termination} with $\Lambda := G$.
 Note that $G$ can be computed in polynomial time, as argued in the proof of Theorem~\ref{thm:parity-qual}.
\qed
\end{proof}

\begin{example}
 We continue Example~\ref{ex:running}, where we have effectively computed $G = \{1,4\}$,
  and thus established that $\Pr[t_1 \text{ is good}] = \Pr[t_4 \text{ is good}] = 1$.
 By Proposition~\ref{prop:AFG} the probabilities $\Pr[t_2 \text{ is good}]$ and $\Pr[t_3 \text{ is good}]$ are given by
  $\Pr[t_2 \models \ms{AF} G]$ and $\Pr[t_3 \models \ms{AF} G]$.
 Proposition~\ref{prop:bp-termination} assures that these probabilities are PPS-expressible;
  in fact they are given by the least nonnegative solution of the equation system
 $[x_2 = \frac12 x_3 + \frac12 x_2 x_3, \ x_3 = \frac23 x_3^2 + \frac13]$,
 which is $x_2 = \frac13$ and $x_3 = \frac12$.
 Hence, we have $\Pr[t_2 \text{ is good}] = \frac13$ and $\Pr[t_3 \text{ is good}] = \frac12$.
\end{example}

\subsubsection{A Lower Bound.}
\newcommand{\Ahit}{\A_\mathit{hit}}
We close the section with a hardness result in terms of the PosSLP problem,
 which asks whether a given straight-line program or,
 equivalently, arithmetic circuit with operations $\mathord{+}$, $\mathord{-}$, $\mathord{\cdot}$, and inputs 0 and~1, and a designated output gate,
 outputs a positive integer or not.
PosSLP is in PSPACE, but known to be in~NP.
The PosSLP problem is a fundamental problem for numerical computation,
 see \cite{AllenderBKM06} for more details.

\newcommand{\ini}{q}
\newcommand{\fin}{r}
For given~$\Gamma$ with $D \in \Gamma$,
 consider the DPTA $\Ahit = (\{\ini,\fin\},\Gamma,a,\delta,c)$ with
 $c(\ini) = 1$ and $c(\fin) = 2$; \
 $\delta(\ini,X,1) = (\ini)$ and $\delta(\ini,X,2) = (\ini,\ini)$ for $X \in \Gamma \setminus \{D\}$; \
 $\delta(\ini,D,1) = (\fin)$ and $\delta(\ini,D,2) = (\fin,\fin)$; \
 $\delta(\fin,X,1) = (\fin)$ and $\delta(\fin,X,2) = (\fin,\fin)$ for $X \in \Gamma$.
Automaton~$\Ahit$ specifies that all branches satisfy the LTL property~$\ms{F} D$,
 i.e., all branches eventually hit~$D$.
Let QUANT-HIT denote the problem to decide whether
 $\Pr_X^\Delta[t \text{ is accepted by~$\Ahit$}] > p$ holds for
  a given branching process~$\Delta = (\Gamma, \mathord{\btran{}}, \mathit{Prob})$ with $X \in \Gamma$
  and a given rational $p \in (0,1)$.
By Theorem~\ref{thm:parity-quant} and Proposition~\ref{prop:PPS-expressible}, QUANT-HIT is in PSPACE.
We have the following proposition:
\newcommand{\stmtprophardness}{
  QUANT-HIT is PosSLP-hard. 
}
\begin{proposition}[see Theorem~5.3 of~\cite{EYstacs05Extended}] \label{prop-hardness}
\stmtprophardness
\end{proposition}

\section{Logic Specifications} \label{sec:logic}

In this section, we propose a logic akin to PCTL, called \emph{probabilistic tree temporal logic},
 to specify the properties of random trees generated from a branching process.
We also present model-checking algorithms for this logic.

\begin{definition}[PTTL]
   \emph{Probabilistic Tree Temporal Logic (PTTL)} formulae over a set~$\Sigma$ of atomic propositions
    are defined by the following grammar:
\begin{align*}
  \phi,\phi' & ::= \top \mid a \mid \neg \phi \mid \phi \wedge \phi' \mid  [\psi]_{\bowtie r}  \\
  \psi & ::= \ms{AX} \phi \mid \ms{EX}\phi \mid \phi\ms{AU}\phi' \mid \phi\ms{EU}\phi' \mid  \phi\ms{AR}\phi' \mid  \phi\ms{ER}\phi'\,,
\end{align*}
where $a\in \Sigma$,  $\mathord{\bowtie} \in\{\mathord{<}, \mathord{\leq}, \mathord{\geq}, \mathord{>}\}$, and $r\in \Q \cap [0,1]$.
If $r \in \{0,1\}$ holds for all subformulae of a PTTL formula~$\phi$, we say that $\phi$ is in the \emph{qualitative fragment} of PTTL.
We use standard abbreviations such as $\bot$ for~$\neg\top$, $\ms{AF}\phi$ for $\top\ms{AU}\phi$, $\ms{EG}\phi$ for $\bot\ms{ER}\phi$, etc.
\end{definition}
For the PTTL semantics we need the notion of a \emph{labelled} branching process,
 which is a branching process $\Delta=(\Gamma, \mathord{\btran{}}, \mathit{Prob})$ extended by a function $\chi: \Gamma \to 2^\Sigma$,
  where $\chi(X)$ indicates which atomic propositions the type~$X$ satisfies.
\begin{definition}[Semantics of PTTL]
\label{def:pttl_sem}
Given a labelled branching process $\Delta=(\Gamma, \mathord{\btran{}}, \mathit{Prob}, \chi)$,
 we inductively define a \emph{satisfaction relation}~$\mathord{\models}$ as follows,
 where $X \in \Gamma$:
\begin{align*}
& X \models \top \\[-1mm]
& X \models a                   && \Leftrightarrow  a\in \chi(X)\\[-1mm]
& X \models \neg \phi           && \Leftrightarrow  X\not\models \phi\\[-1mm]
& X \models \phi \wedge \phi'   && \Leftrightarrow  X \models \phi \text{ and } X \models \phi'\\[-1mm]
& X \models [\psi]_{\bowtie r}  && \Leftrightarrow  \Pr_X^\Delta[t_X \models \psi]\bowtie r\\
& t \models \ms{AX}\phi         && \Leftrightarrow  \text{for all branches } u_0u_1\cdots \text{ of~$t$ we have } \ell(u_1)\models \phi\\[-1mm]
& t \models \phi\ms{AU}\phi'    && \Leftrightarrow  \text{for all branches } u_0u_1\cdots \text{ of~$t$ there exists } i\in \mb{N} \text{ with} \\[-1.5mm]
&                               && \qquad \ell(u_i) \models \phi' \text{ and for all $0\leq j<i$ we have } \ell(u_j) \models \phi\\[-1mm]
& t \models \phi\ms{AR}\phi'    && \Leftrightarrow  \text{for all branches } u_0u_1\cdots \text{ of~$t$ and for all $i \in \N$ we have} \\[-1.5mm]
&                               && \qquad \ell(u_i) \models \phi' \text{ or there exists $0 \le j < i$ with } \ell(u_j) \models \phi
\end{align*}
The modalities $\ms{EX}$, $\ms{EU}$ and $\ms{ER}$ are defined similarly, with ``for all branches'' replaced by ``there exists a branch''.
\end{definition}
We now present the model checking algorithm.
The algorithm shares its basic structure with the well-known algorithm for (P)CTL and finite (probabilistic) transition systems.
Given a PTTL formula~$\phi$, the algorithm recursively evaluates the truth values of the PTTL subformulae $\psi$ of~$\phi$ for all types.
The boolean operators can be dealt with as in the CTL algorithm.
Hence, it suffices to examine formulae of the form $[\psi]_{\bowtie r}$.
Observe that we have
$\ms{EX}\phi \equiv \neg \ms{AX} \neg \phi$ and
 $\phi\ms{ER}\phi' \equiv  \neg (\neg\phi\ms{AU}\neg\phi')$ and $\phi\ms{EU}\phi' \equiv \neg (\neg\phi\ms{AR}\neg\phi')$ and
\[X\models [\neg \phi]_{\bowtie r}\mbox{ if and only if }X\models [\phi]_{\bar{\bowtie} 1-r}\,,\]
where $\mathord{\bar{\bowtie}} \in \{\mathord{\ge}, \mathord{>}, \mathord{<}, \mathord{\le}\}$ is the complement operator of
      $\mathord{\bowtie}       \in \{\mathord{<}, \mathord{\leq}, \mathord{\geq}, \mathord{>}\}$.
Hence, it suffices to deal with the following three cases:
 (i) $X\models [\ms{AX} \phi]_{\bowtie r}$;
 (ii)  $X\models [\phi \ms{AU} \psi]_{\bowtie r}$;
 (iii) $X\models [\phi \ms{AR} \psi]_{\bowtie r}$.
We assume in the following case distinction that the algorithm has already computed the truth values of the subformulae $\phi, \psi$.
One could now construct a suitable DPTA for each of the cases (i)--(iii), and proceed according to the machinery of Section~\ref{sec:parity}.
Instead we present in the following a more direct and more efficient algorithm
 which takes advantage of the special shape of the linear-time operators $\ms{X}$, $\ms{U}$ and~$\ms{R}$.

\noindent\emph{Case (i):}
We have
$\displaystyle \Pr[t_X \models \ms{AX} \phi]
 = \mathop{\sum_{X \btran{p} Y_1\dots Y_k}}_{Y_1,\dots,Y_k \models \phi} p$,
 which is easy to compute.
So one can decide in polynomial time whether $X \models [\ms{AX}\phi]_{\bowtie r}$.

\noindent\emph{Case (ii):}
We reduce the check of the $\phi \ms{AU} \psi$ modality to a check of~$\ms{AF}$.
To this end, we define a branching process $\Delta' = (\Gamma \times \{0,\frac12,1\}, \mathord{\btran{}'}, \mathit{Prob}')$
 which tracks the ``status'' of $\phi \ms{AU} \psi$.
We define~$\Delta'$ in terms of an auxiliary function
 $f_{\phi,\psi}: \Gamma \to \{0,\frac12,1\}$ with
$f_{\phi,\psi}(Y) = 0$ if $Y \models \neg \phi \ \land \ \neg \psi$,
$f_{\phi,\psi}(Y) = \frac12$ if $Y \models \phi \ \land \ \neg \psi$, and
$f_{\phi,\psi}(Y) = 1$ if $Y \models \psi$.
For any rule $X \btran{p} Y_1\dots Y_k$ in $\Delta$, there are three
corresponding rules in $\Delta'$, namely
 $(X,0) \btran{p} (Y_1,0)\dots(Y_k,0)$,
 $(X,1) \btran{p} (Y_1,1)\dots(Y_k,1)$, and
 $(X,\frac12) \btran{p} (Y_1,f_{\phi,\psi}(Y_1))\dots(Y_k,f_{\phi,\psi}(Y_k))$.
By this construction we achieve $\Pr_X^\Delta[t_X \models \phi \ms{AU} \psi] = \Pr_{X'}^{\Delta'}[t_{X'} \models \ms{AF} \Lambda]$ for $X' = (X,f_{\phi,\psi}(X))$
 and $\Lambda := \Gamma \times \{1\}$.
Hence, using Propositions \ref{prop:PPS-expressible} and~\ref{prop:bp-termination} we obtain
 that whether $X \models [\phi \ms{AU} \psi]_{\bowtie r}$ holds is decidable in PSPACE;
 and in polynomial time for $r \in \{0,1\}$.

\noindent\emph{Case (iii):}
Similarly to case~(ii) we reduce the check of $\phi \ms{AR} \psi$ to a check of~$\ms{AG}$.
This time we define~$\Delta' = (\Gamma \times \{0,\frac12,1\}, \mathord{\btran{}'}, \mathit{Prob}')$ in terms of an auxiliary function
 $g_{\phi,\psi}: \Gamma \to \{0,\frac12,1\}$ with
$g_{\phi,\psi}(Y) = 0$ if $Y \models \neg \psi$,
$g_{\phi,\psi}(Y) = \frac12$ if $Y \models \neg \phi \ \land \ \psi$,
$g_{\phi,\psi}(Y) = 1$ if $Y \models \phi \ \land \ \psi$.
The rules of~$\Delta'$ are defined as in case~(ii), except that $f_{\phi,\psi}$ is replaced with $g_{\phi,\psi}$.
By this construction we achieve $\Pr_X^\Delta[t_X \models \phi \ms{AR} \psi] = \Pr_{X'}^{\Delta'}[t_{X'} \models \ms{AG} \Lambda]$ for $X' = (X,g_{\phi,\psi}(X))$
 and $\Lambda := \Gamma \times \{\frac12,1\}$.
The following lemma allows to express this probability in terms of $\ms{AF}$ instead of~$\ms{AG}$:
\newcommand{\stmtlemAGAF}{
 Let $\Delta=(\Gamma, \mathord{\btran{}}, \mathit{Prob})$ be a branching process.
 Let $\Lambda \subseteq \Gamma$ such that no type in~$\Lambda$ is reachable from any type in $\Gamma \setminus \Lambda$.
 Define $G := \{Y \in \Lambda \mid \text{all types reachable from~$Y$ are in~$\Lambda$}\}$.
 Let $X \in \Gamma$.
 Then $\Pr[t_X \models \ms{AG} \Lambda] = \Pr[t_X \models \ms{AF} G]$.
}
\begin{lemma} \label{lem:AG-AF}
 \stmtlemAGAF
\end{lemma}
To summarize case~(iii): we have reduced $\ms{AR}$ to~$\ms{AG}$ and then $\ms{AG}$ to~$\ms{AF}$.
Hence, using Propositions \ref{prop:PPS-expressible} and~\ref{prop:bp-termination} we obtain
 that whether $X \models [\phi \ms{AR} \psi]_{\bowtie r}$ holds is decidable in PSPACE;
 and in polynomial time for $r \in \{0,1\}$.

As the overall algorithm computes the truth values of the subformulae recursively,
 we have proved the following theorem:
\begin{theorem} \label{thm:logic}
 Model checking branching processes against PTTL is in~{\rm PSPACE}.
 Model checking branching processes against the qualitative fragment of PTTL is in~{\rm P}.
\end{theorem}

\section{Conclusions and Future Work} \label{sec:conclusions}
Branching processes are a basic formalism for modelling probabilistic parallel programs with dynamic process creation.
This paper is the first to consider the verification of branching processes,
We have shown how to model check specifications given in terms of deterministic parity automata,
 a problem that unifies and strictly generalises linear-time model-checking problems for Markov chains and for (nonprobabilistic) nondeterministic transition systems.
We have also provided model-checking algorithms for a new logic, PTTL, suitable for specifying probabilistic properties of random trees.
To obtain these results we have provided reductions to computing the probability of hitting ``good'' states along all branches.

Future research in this area should involve:
\begin{itemize}
\item
 the complexity of the problem where the specification is an LTL formula required to hold on all branches;
\item
 the problem where deterministic parity automata are replaced by other tree specification formalisms,
  such as CTL (or CTL$^*$) formulae;
\item
 extending the model-checking algorithms to accommodate the synchronisation and communication features of probabilistic split-join systems.
\end{itemize}
It seems that at least the latter two problems require additional techniques,
 as the children of a node in the branching process can no longer be treated independently.

\medskip
\noindent \textbf{Acknowledgements.} We thank anonymous reviewers for their valuable feedback.

\bibliographystyle{plain} 
\bibliography{db}

\begin{thebibliography}{10}

\bibitem{AllenderBKM06}
E.~Allender, P.~B{\"u}rgisser, J.~Kjeldgaard-Pedersen, and P.~B. Miltersen.
\newblock On the complexity of numerical analysis.
\newblock In {\em IEEE Conference on Computational Complexity}, pages 331--339,
  2006.

\bibitem{AthreyaNey72}
K.B. Athreya and P.E. Ney.
\newblock {\em Branching Processes}.
\newblock Springer, 1972.

\bibitem{Can88}
J.~Canny.
\newblock Some algebraic and geometric computations in {PSPACE}.
\newblock In {\em STOC'88}, pages 460--467, 1988.

\bibitem{CourcYann95}
C.~Courcoubetis and M.~Yannakakis.
\newblock The complexity of probabilistic verification.
\newblock {\em Journal of the ACM}, 42:857--907, 1995.

\bibitem{EGK10:stacs}
J.~Esparza, A.~Gaiser, and S.~Kiefer.
\newblock Computing least fixed points of probabilistic systems of polynomials.
\newblock In {\em Proceedings of STACS}, pages 359--370, 2010.

\bibitem{EKL10:SICOMP}
J.~Esparza, S.~Kiefer, and M.~Luttenberger.
\newblock Computing the least fixed point of positive polynomial systems.
\newblock {\em SIAM Journal on Computing}, 39(6):2282--2335, 2010.

\bibitem{EKM04}
J.~Esparza, A.~Ku\v{c}era, and R.~Mayr.
\newblock Model checking probabilistic pushdown automata.
\newblock In {\em LICS'04}, pages 12--21. IEEE, 2004.

\bibitem{EtessamiSY12:poly-time}
K.~Etessami, A.~Stewart, and M.~Yannakakis.
\newblock Polynomial-time algorithms for multi-type branching processes and
  stochastic context-free grammars.
\newblock In {\em Proceedings of STOC}, pages 579--588, 2012.

\bibitem{EYstacs05Extended}
K.~Etessami and M.~Yannakakis.
\newblock Recursive {M}arkov chains, stochastic grammars, and monotone systems
  of nonlinear equations.
\newblock {\em Journal of the ACM}, 56(1):1--66, 2009.

\bibitem{EY12journal}
K.~Etessami and M.~Yannakakis.
\newblock Model checking of recursive probabilistic systems.
\newblock {\em ACM Transactions on Computational Logic}, 13(2), 2012.
\newblock To appear.

\bibitem{Harris63}
T.E. Harris.
\newblock {\em The Theory of Branching Processes}.
\newblock Springer, 1963.

\bibitem{KW11:TACAS}
S.~Kiefer and D.~Wojtczak.
\newblock On probabilistic parallel programs with process creation and
  synchronisation.
\newblock In {\em Proceedings of TACAS}, volume 6605 of {\em LNCS}, pages
  296--310. Springer, 2011.

\bibitem{Renegar92}
J.~Renegar.
\newblock On the computational complexity and geometry of the first-order
  theory of the reals. {Parts I--III}.
\newblock {\em Journal of Symbolic Computation}, 13(3):255--352, 1992.

\bibitem{Thomas97handbook}
W.~Thomas.
\newblock Languages, automata, and logic.
\newblock In G.~Rozenberg and A.~Salomaa, editors, {\em Handbook of Formal
  Languages}, volume 3, Beyond Words, pages 389--455. Springer, 1997.

\bibitem{Vardi99}
M.Y. Vardi.
\newblock Probabilistic linear-time model checking: An overview of the
  automata-theoretic approach.
\newblock In {\em Formal Methods for Real-Time and Probabilistic Systems},
  volume 1601 of {\em LNCS}, pages 265--276. Springer, 1999.

\end{thebibliography}

\iftechrep{
\newpage
\appendix

\section{Omitted Proofs}

\subsection{Proof of Lemma~\ref{lem:cut-or-inf}}

\begin{qlemma}{\ref{lem:cut-or-inf}}
\stmtlemcutorinf
\end{qlemma}

\begin{proof}
We first show $C \cap I = \emptyset$.
Let $t \in I$.
Consider the set~$W$ of those nodes~$w$ in~$t$ such that the path to~$w$
 is neither $X$-closing nor $X$-failing and the same holds for all prefixes of this path.
Since the parents of all nodes in the infinite set~$\Clos_X(t)$ are in~$W$, the set~$W$ is infinite as well.
It follows from K\H{o}nig's lemma that $W$ contains an infinite path, hence a branch in~$t$ that does not hit $\Clos_X(t) \cup \Fail_X(t)$.
So $t \not\in C$.
Thus $C \cap I = \emptyset$.

It remains to show $\Pr[C \cup I] = 1$.
For the proof we follow a particular pattern which we will use several times in this paper:
 we describe a ``procedure which unfolds a tree in stages''.
Such a ``procedure'' takes a tree~$t_X \in \trees{\Delta}_X$ randomly generated by~$\Delta$,
 and inspects finite prefixes of~$t_X$ according to the procedure's pseudocode.
In each step, the procedure accumulates ``observations'' on~$t_X$,
 e.g., on whether or not a node with certain properties has been visited.
Denote by~$O_i(t_X)$ the sequence of observations the procedure makes on~$t_X$ in the first $i$ steps.
For each observation sequence~$o$, denote by~$E(o,i)$ the event that the procedure observes~$o$ in the first $i$ steps, i.e.,
 $E(o,i) = \{t \in \trees{\Delta}_X \mid O_i(t_X) = o\}$.
Any such event $E(o,i)$ is \emph{measurable}, as the procedure looks only at finite prefixes of~$t_X$.
It follows that events such as ``the procedure does not terminate''
 and ``the procedure visits at least $n$ nodes in $\Clos_X(t_X)$''
 are measurable as well.
We follow this pattern in the rest of this proof and give some more details at the end.
The other ``procedural'' proofs in this paper can be treated analogously.

Let $r$ denote the root of~$t := t_X$.
Consider the following procedure which unfolds~$t$ in stages:
\begin{itemize}
\item[1.]
 Initialise a set~$S$ with $S := \{ r \}$.
\item[2.]
 Pick\footnote{To resolve the ``nondeterminism'', we can pick, e.g., the lexicographically smallest node.}
 and remove from~$S$ a node~$u$ and unfold $|\Gamma|$ levels of~$t_u$.
 Let $L$ denote the set of the new ``leaves'', i.e., those descendants of~$u$ that have distance~$|\Gamma|$ from~$u$.
\item[3.]
 Remove from~$L$ those nodes~$w$ that have a (proper or improper) ancestor~$v$ with $v \in \Clos_X(t) \cup \Fail_X(t)$.
 Add the remaining nodes in~$L$ to~$S$.
\item[4.]
 If $S$ is empty, then report ``$t \in C$'' and terminate.
 Otherwise goto 2.
\end{itemize}
If the procedure terminates, it correctly reports ``$t \in C$''.
If it does not terminate, then almost surely $t \in I$, because there is $p>0$ such that in each execution of step~2.\
 the probability of reaching at least one new node in~$\Clos_X(t)$ is at least~$p$.
In other words, the probability of nontermination equals the probability of~$I$.
Hence $\Pr[C \cup I] = 1$.

\newcommand{\Non}{\mathit{Non}}
Let us give some more details on why the probability of nontermination in fact equals the probability of~$I$.
As $C \cap I = \emptyset$, the event $I$ implies nontermination.
So it suffices to argue that $\Pr[\Non \cap F] = 0$,
 where $\Non$ denotes nontermination and $F := \{t_X \mid \Clos_X(t_X) \text{ is finite}\}$.
Consider the event $E_{i,n}$ that after $i$ iterations the procedure has not yet terminated and the number of visited $\Clos_X(t_X)$-nodes is less than~$n$.
As $\Non \cap F = \bigcup_{n\in \N} \bigcap_{i\in \N} E_{i,n}$,
 it suffices to argue that for each $n \in \N$ we have $\lim_{i\to\infty} \Pr[E_{i,n}] = 0$.
Fix an arbitrary $n \in \N$.
In each iteration, there is a positive probability that the procedure terminates or hits at least one new node in $\Clos_X(t_X)$.
This probability is bounded below by some $p>0$:
 take $p$ as the minimal probability over all types $Y \in \Gamma \setminus N_X$ such that a tree rooted at~$Y$ has an $X$-closing path of length~$|\Gamma|$.
Hence, in $n$ procedure iterations the probability of termination or hitting at least $n$ new nodes in $\Clos_X(t_X)$ is at least $q := p^n > 0$.
It follows that $\Pr[E_{i+n,n}] \le (1-q) \Pr[E_{i,n}]$.
Hence we have $\lim_{i\to \infty} \Pr[E_{i,n}] = 0$, as desired.
\qed
\end{proof}

\subsection{Proof of Lemma~\ref{lem:cut-or-dirty}}

We show the following lemma from the main body of the paper:
\begin{qlemma}{\ref{lem:cut-or-dirty}}
\lemcutordirty
\end{qlemma}
\begin{proof}
The equality $F \cap H = \emptyset$ is obvious, so it suffices to show that $\Pr[F \cup H] = 1$.
Let $r$ denote the root of~$t := t_X$.
Consider the following procedure which unfolds~$t$ in stages:
\begin{itemize}
\item[1.] 
 Initialise a set~$S$ with $S := \{ r \}$.
\item[2.]
 Pick and remove from~$S$ a node~$u$ and unfold~$t_u$
  until all ``leaves'' of~$t_u$ are in $\Clos_X(t_u) \cup \Fail_X(t_u)$.
  (Note that this step may not terminate.)
\item[3.]
 Add to~$S$ all those ``leaves'' of~$t_u$ that are in~$\Clos_X(t_u)$.
\item[4.]
 If $S$ is empty, then report ``$t \in F$'' and terminate.
 Otherwise goto 2.
\end{itemize}
If the procedure terminates, it correctly reports ``$t \in F$''.
Using the event~$I$ from Lemma~\ref{lem:cut-or-inf} we distinguish between two cases:
\begin{description}
 \item[(a)] \protect{\normalfont $\Pr[I] = 0$}:
  Lemma~\ref{lem:cut-or-inf} implies that step~2.\ of the above procedure terminates almost surely in every iteration.
  If the overall procedure does not terminate, consider the set~$M$ of those nodes that are in~$S$ at some point during the execution of the procedure.
  This set~$M$ is infinite.
  Then it follows from K\H{o}nig's lemma that $t$ has a branch with infinitely many nodes in~$M$.
  In this branch, any two distinct nodes in~$M$ define an $X$-closing path.
  Hence, the branch is an $X$-branch, so $t \in H$.
 \item[(b)] \protect{\normalfont $\Pr[I] > 0$}:
  Let $a := \Pr[I] > 0$ and choose $k\in \N$ such that $k \cdot a > 1$.
  Consider the \emph{$(X,k)$-skeleton} of the branching process, i.e., the branching process with a single type~$X$ and rules
   \[ X \btran{p_i} \underbrace{X \cdots X}_{\text{$i$ times}} \qquad  \text{for } i \in \{0, 1, \ldots, k\}\]
  where, for $i \le k-1$, the probability~$p_i$ is the probability that a random tree $t_X$ satisfies $|\Clos_X(t_X)| = i$ and
   $p_k$ is the probability that $|\Clos_X(t_X)| \ge k$.
  We claim that the tree generated by the $(X,k)$-skeleton is infinite with positive probability.
  We argue by comparing with a ``smaller'' branching process:
  It is a fact in the theory of branching processes that $k \cdot a > 1$ implies that the tree generated by the branching process with the rules
   \[ X \btran{a} \underbrace{X \cdots X}_{\text{$k$ times}} \qquad \text{and} \qquad X \btran{1-a} \varepsilon \]
  is infinite with positive probability.
  Hence the same holds for the $(X,k)$-skeleton.

\newcommand{\wt}{\widetilde{t}}%
  Now assume the procedure above does not terminate, then, almost surely, in one of the executions of step~2.\
   the set $\Clos_X(t_u)$ is infinite (implying that step~2.\ does not terminate).
  Let $\Clos_X(t_u) = \{v_1, v_2, \ldots\}$.
  Each of the $v_i$ can be regarded as the root of a tree generated by the $(X,k)$-skeleton:
   regard the elements of~$\Clos_X(t_{v_i}) = \{w_{i1}, w_{i2}, \ldots\}$
   (or the first~$k$ according to an arbitrary order, if there are more than~$k$) as \emph{direct} children of~$v_i$,
   and the elements of~$\Clos_X(t_{w_{ij}})$ as \emph{direct} children of~$w_{ij}$ etc.
  In this view, it follows from the previous discussion that, for all $i$, the tree~$\wt_{v_i}$ is infinite with positive probability,
   where by $\wt_{v_i}$ we mean the tree obtained from~$t_{v_i}$ by contracting as described above.
  As a result, there is almost surely an~$i$ such that $\wt_{v_i}$ is infinite.
  Since the $(X,k)$-skeleton is finitely branching, it follows that $\wt_{v_i}$ has an infinite branch.
  By the definition of~$\Clos_X$ this branch corresponds to an $X$-branch in $t_{v_i}$.
  As the path from the root of~$t$ to~$v_{i}$ is $X$-closing, it follows $t \in H$.
\end{description}
We conclude that the probability of nontermination equals the probability of~$H$.
Hence $\Pr[F \cup H] = 1$.
\qed
\end{proof}

\subsection{Proof of Proposition~\ref{prop:AFG}}

We prove the following proposition from the main body of the paper:

\begin{qproposition}{\ref{prop:AFG}}
\stmtpropAFG
\end{qproposition}
\begin{proof}
Define the events $A := \{t_X \mid t_X \models \mathsf{AF} G\}$ and $B := \{t_X \mid t_X \text{ is not good}\}$.
We need to show that $\Pr[A] + \Pr[B] = 1$.
First we show that $\Pr[A \cap B] = 0$.
Assume $t_X \in A$, so all branches in~$t_X$ go through a node after which all reachable types are clean.
Since the set of nodes in a tree is countable, it follows that almost surely no subtree of~$t_X$ has a $Y$-branch for odd~$Y$;
 i.e., $t_X$ is almost surely good.
Hence $\Pr[A \cap B] = 0$.

Now it suffices to show that $\Pr[A \cup B] = 1$.
Let $r$ denote the root of~$t := t_X$.
Consider the following procedure which unfolds~$t$ in stages:
\begin{itemize}
\item[1.] 
 Initialise a set~$S$ with $S := \{ r \}$.
\item[2.]
 \begin{itemize}
  \item[(a)]
   If $S$ contains a node~$u$ of unclean type, say, $Y$, then remove $u$ from~$S$.
   Unfold~$t_u$ until all ``leaves'' of~$t_u$ have $N_Y$-type.
   (Note that this step may not terminate.)
  \item[(b)]
   Otherwise pick and remove from~$S$ a node~$u$ of clean type.
   Unfold~$t_u$ until either all ``leaves'' of~$t_u$ have a $G$-type or at least one ``leaf'' of~$t_u$ has an unclean type.
 \end{itemize}
\item[3.]
  Add all ``leaves'' of~$t_u$ to~$S$ except those that have $G$-type.
\item[4.]
 If $S$ is empty, then report ``$t \in A$'' and terminate.
 Otherwise goto 2.
\end{itemize}
If the procedure terminates, it correctly reports ``$t \in A$''.
If an execution of step~2.~(a) does not terminate, then, by Lemma~\ref{lem:cut-or-dirty}, the tree~$t_u$ almost surely has a $Y$-branch,
 implying that $t \in B$.
In each execution, step~2.~(b) terminates almost surely, as there is $p>0$ such that any node with non-$G$-type reaches an unclean type with probability at least~$p$.
The outer loop (``otherwise goto 2.'') is almost surely executed only finitely often:
 this is because whenever step~2.~(a) is executed, there is a positive probability of nontermination in step~2.~(a);
 and whenever step~2.~(b) is executed, there is a positive probability that this is the last execution of step~2.~(b),
  as there is a positive probability of reaching a ``leaf'' of unclean type,
  which, again with positive probability, results in nontermination during the following execution of step~2.~(a).

We conclude that the probability that the above procedure does not terminate equals the probability of~$B$.
Hence $\Pr[A \cup B] = 1$.
\qed
\end{proof}

\subsection{Proof of Proposition~\ref{prop-hardness}}

We prove the following proposition from the main body of the paper:

\begin{qproposition}{\ref{prop-hardness}}
\stmtprophardness
\end{qproposition}
\begin{proof}
The proof is immediate from Theorem~5.3 of~\cite{EYstacs05Extended},
 as the ``quantitative termination'' problem studied there for so-called 1-exit-RMCs
  corresponds exactly to the $\Ahit$-specification for branching processes.
However, we remark that if the automaton~$\A$ is part of the input,
 the problem to decide whether $\Pr_X^\Delta[t \text{ is accepted by~$\A$}] > p$ holds
 cannot be translated to an RMC problem, at least not in a straightforward way.
\qed
\end{proof}

\subsection{Proof of Lemma~\ref{lem:AG-AF}}

\begin{qlemma}{\ref{lem:AG-AF}}
 \stmtlemAGAF
\end{qlemma}
\begin{proof}
Define the events $A := \{t_X \mid t_X \models \mathsf{AF} G\}$ and $B := \{t_X \mid t_X \not\models \ms{AG} \Lambda\}$.
We need to show that $\Pr[A] + \Pr[B] = 1$.
First we show that $A \cap B = \emptyset$.
Observe that $B$ is the event that $t_X$ has a non-$\Lambda$ node, say~$v$.
By the assumptions of the lemma, neither an ancestor nor a descendant of~$v$ nor $v$ itself can have a $G$-label.
So $t_X \not\in A$.
Hence $\Pr[A \cap B] = 0$.

Now it suffices to show that $\Pr[A \cup B] = 1$.
Let $r$ denote the root of~$t := t_X$.
Consider the following procedure which unfolds~$t$ in stages:
\begin{itemize}
\item[1.] If $X \in G$, report ``$t_X \in A$'' and terminate.
 Otherwise, initialise a set~$S$ with $S := \{ r \}$.
\item[2.]
 Pick and remove from~$S$ a node~$u$ and unfold $|\Gamma|$ levels of~$t_u$.
 Let $L$ denote the set of the new ``leaves'', i.e., those descendants of~$u$ that have distance~$|\Gamma|$ from~$u$.
\item[3.]
 If there is a node~$w$ in~$L$ with $\ell(w) \not\in \Lambda$,
  then report ``$t \in B$'' and terminate.
 Otherwise, remove from~$L$ those nodes~$w$ with $\ell(w) \in G$ and add the remaining nodes in~$L$ to~$S$.
\item[4.]
 If $S$ is empty, then report ``$t \in A$'' and terminate.
 Otherwise goto 2.
\end{itemize}
Clearly, if the procedure terminates, the reported result is correct.
It remains to show that the procedure terminates almost surely.
Observe that all nodes that are in~$S$ at some point have a non-$G$ label.
So in each execution of step~2.\ there is a nonzero probability of hitting a non-$\Lambda$ node, which forces termination in step~3.
\qed
\end{proof}

}
{}
\end{document}